%% file: paper-proof-complexity.tex
\documentclass[11pt]{article}
\usepackage{amssymb,amsmath}

\usepackage{fullpage}

\usepackage{times}
\usepackage{helvet}
\usepackage{courier}
\usepackage{amssymb,amsmath}
\usepackage{comment}
\usepackage{graphicx}

\usepackage{mathabx}

\usepackage{fancyheadings}

\usepackage{amsmath}
\usepackage{comment}
\usepackage{amssymb}

\newcommand{\confversion}[1]{}
\newcommand{\fullversion}[1]{#1}

\newcommand{\ignore}[1]{}

\newenvironment{itemized}{
\begin{itemize} 
  \setlength{\itemsep}{0pt}
  \setlength{\parskip}{0pt}
  \setlength{\parsep}{0pt}
}
{\end{itemize}}

\begin{document}

\title{Proof Complexity Modulo the Polynomial Hierarchy:\\
Understanding Alternation as a Source of Hardness}

\author{Hubie Chen\\
Universidad del Pa\'{i}s Vasco,
E-20018 San Sebasti\'{a}n,
Spain\\
\emph{and}
IKERBASQUE, Basque Foundation for Science,
E-48011 Bilbao,
Spain
}

\date{ } 

\maketitle

\begin{abstract}
\begin{quote}
We present and study a framework in which
one can present alternation-based lower bounds
on proof length in proof systems for quantified Boolean formulas.
A key notion in this framework is that of
\emph{proof system ensemble}, which is 
(essentially) a sequence of proof systems
where, for each, proof checking can be performed in 
the polynomial hierarchy.
We introduce a proof system ensemble
called \emph{relaxing QU-res} which is based on the
established proof system \emph{QU-resolution}.
Our main results include an exponential separation of 
the tree-like and general versions of relaxing QU-res,
and an exponential lower bound for relaxing QU-res;
these are analogs of 
classical results in propositional proof complexity.
\end{quote}
\end{abstract}
\fullversion{ \setcounter{page}{1}}
\confversion{
\setcounter{page}{0}
\newpage
}


\newtheorem{theorem}{Theorem}[section]
\newtheorem{conjecture}[theorem]{Conjecture}
\newtheorem{corollary}[theorem]{Corollary}
\newtheorem{proposition}[theorem]{Proposition}
\newtheorem{prop}[theorem]{Proposition}
\newtheorem{lemma}[theorem]{Lemma}
\newtheorem{remark}[theorem]{Remark}
\newtheorem{exercisecore}[theorem]{Exercise}
\newtheorem{examplecore}[theorem]{Example}

\newenvironment{example}
  {\begin{examplecore}\rm}
  {\hfill $\Box$\end{examplecore}}

\newenvironment{exercise}
  {\begin{exercisecore}\rm}
  {\hfill $\Box$\end{exercisecore}}

\newenvironment{proof}{\noindent\textbf{Proof\/}.}{$\Box$ \vspace{1mm}}

\newtheorem{researchq}{Research Question}

\newtheorem{newremarkcore}[theorem]{Remark}

\newenvironment{newremark}
  {\begin{newremarkcore}\rm}
  {\end{newremarkcore}}

\newtheorem{definitioncore}[theorem]{Definition}

\newenvironment{definition}
  {\begin{definitioncore}\rm}
  {\end{definitioncore}}

\newcommand{\hubie}[1]{{\bf Hubie:} #1 {\bf End.}}

\newcommand{\ppequiv}{\mathsf{PPEQ}}
\newcommand{\eq}{\mathsf{EQ}}
\newcommand{\iso}{\mathsf{ISO}}
\newcommand{\ppeq}{\ppequiv}
\newcommand{\ppiso}{\mathsf{PPISO}}
\newcommand{\boolppiso}{\mathsf{BOOL}\mbox{-}\mathsf{PPISO}}
\newcommand{\csp}{\mathsf{CSP}}
\newcommand{\gi}{\mathsf{GI}}
\newcommand{\ci}{\mathsf{CI}}
\newcommand{\relb}{\mathbf{B}}
\newcommand{\alga}{\mathbb{A}}
\newcommand{\algb}{\mathbb{B}}
\newcommand{\algab}{\mathbb{A}_{\relb}}

\newcommand{\idemp}{I}

\newcommand{\varv}{\mathcal{V}}
\newcommand{\variety}{\mathcal{V}}
\newcommand{\false}{\mathsf{false}}
\newcommand{\true}{\mathsf{true}}
\newcommand{\pol}{\mathsf{Pol}}
\newcommand{\inv}{\mathsf{Inv}}
\newcommand{\cc}{\mathcal{C}}
\newcommand{\alg}{\mathsf{Alg}}
\newcommand{\pitwo}{\Pi_2^p}
\newcommand{\sigmatwo}{\Sigma_2^p}
\newcommand{\pithree}{\Pi_3^p}
\newcommand{\sigmathree}{\Sigma_3^p}
\newcommand{\Sigmap}{\Sigma^p}
\newcommand{\Pip}{\Pi^p}

\newcommand{\fancyg}{\mathcal{G}}
\newcommand{\tw}{\mathsf{tw}}

\newcommand{\qc}{\mathsf{QC\mbox{-}MC}}
\newcommand{\rqc}{\mathsf{RQC\mbox{-}MC}}

\newcommand{\qcfo}{\mathrm{QCFO}}
\newcommand{\qcfofk}{\qcfo_{\forall}^k}
\newcommand{\qcfoek}{\qcfo_{\exists}^k}

\newcommand{\fo}{\mathrm{FO}}

\newcommand{\tup}[1]{\overline{#1}}

\newcommand{\nn}{\mathsf{nn}}
\newcommand{\bush}{\mathsf{bush}}
\newcommand{\width}{\mathsf{width}}

\newcommand{\un}{N^{\forall}}
\newcommand{\en}{N^{\exists}}

\newcommand{\ord}{\tup{u}}
\newcommand{\ordp}[1]{\tup{#1}}

\newcommand{\gc}{G^{-C}}

\newcommand{\prp}{\vec{P}}
\newcommand{\prpp}{\vec{P'}}

\newcommand{\qed}{}

\newcommand{\assign}{\mathrm{assign}}
\newcommand{\clause}{\mathrm{clause}}
\newcommand{\dom}{\mathrm{dom}}
\newcommand{\vars}{\mathrm{vars}}
\newcommand{\last}{\mathrm{last}}

\newcommand{\nats}{\mathbb{N}}

\newcommand{\res}{\upharpoonright}

\newcommand{\overnot}[1]{\overline{#1}}

\newcommand{\fancyf}{\mathcal{F}}


\input{framework.tex}

\input{technical.tex}

\newpage

\bibliographystyle{plain}

\bibliography{../../hubiebib}

\input{appendix.tex}

\end{document}

%% file: framework.tex
\section{Introduction} 
\label{sect:introduction}

{\bf Background.}
Traditionally, the area of \emph{propositional proof complexity}
studies proof length in propositional proof systems
for certifying the unsatisfiability of 
instances of the \emph{SAT problem}, which instances are
quantifier-free propositional 
formulas~\cite{CookReckhow74-lengthsofproofs,BeamePitassi98-survey,Segerlind07-survey}.
This line of study is supported by multiple motivations;
let us highlight a few.
First, while satisfiable formulas can be easily certified
by a satisfying assignment,
it is also natural to desire efficiently verifiable proofs
for unsatisfiable formulas (for instance, to
check that a SAT algorithm judged unsatisfiability correctly);
understanding whether and when proof systems have
succinct proofs is a prime concern of this area.
Relatedly, \emph{SAT algorithms} for deciding the SAT problem
can be typically shown to implicitly generate proofs 
in a proof system, and thus insight into proof length
in the resulting proof system can be used to gain insight
into the running-time behavior of SAT algorithms
(see for example the discussions in~\cite{BeameKautzSabharwal04-clauselearning,AtseriasFichteThurley11-learning}).
In addition, the question of whether or not there are
proof systems admitting \emph{polynomially bounded proofs}
is (when formalized) equivalent to the question of whether or not
NP is equal to coNP~\cite{CookReckhow74-lengthsofproofs}, 
and one can thus suggest that
studying proof length in propositional proof systems
sheds light on the relationship between these two complexity classes.

Over recent years, researchers have devoted increasing
attention to methods for solving the \emph{QBF problem},
a generalization of the SAT problem 
and a canonical PSPACE-complete problem; an instance of this problem
is a propositional formula where
each variable is either existentially or universally quantified.
(\emph{QBF} is short for \emph{quantified Boolean formula}.)
It is often suggested that the move to studying this
more general problem is based on advances
in the efficacy of SAT algorithms
(see for example~\cite{YuMalik05-validating}).
As reinforces this suggestion, 
let us point out that one can find 
QBF solution techniques which use SAT algorithms as
black-box, primitive components, and hence which
arguably conceive of and treat the SAT problem as feasibly solvable.
For instance, sKizzo, a QBF solver dating back to 2005, 
would convert the QBF being processed to a SAT instance
and then call a SAT solver, whenever 
this was \emph{affordable}~\cite{Benedetti05-skizzo}.
As another example, a different QBF solver
which extensively calls a SAT solver during a backtrack-style
search
was developed and studied~\cite{SamulowitzBacchus05-usingsat}.

The rise in the study of the QBF problem 
has resulted in the identification of a number of
core algorithmic techniques and corresponding proof systems
that aim to capture these 
 (see for example~\cite{BuningKarpinskiFlogel95-resolution,VanGelder12-contributions,EglyLonsingWidl13-ldresolution,JanotaMarquesSilva13-expansions,BalabanovWidlJiang14-qbfcomplexities,HeuleSeidlBiere14-unified,BeyersdorffChewJanota14-unification,BeyersdorffChewJanota15-qbf-calculi} 
 and the references therein).
We refer to these proof systems
as \emph{QBF proof systems};
they can be used as a basis for certifying a decision 
for a QBF instance.
One can motivate the study of QBF proof systems
in much the same way that the study of propositional proof systems
has been motivated; 
hence, these QBF proof systems
would seem to suggest
a new chapter in the study of proof complexity,
and a new domain for the existing lines of inquiry thereof.

However, one is immediately confronted with a dilemma
upon inspecting
the very basic question of whether or not a typical QBF
proof system requires long (exponentially sized) proofs---again,
a primary type of question
in traditional proof complexity.
As an example, let us discuss
\emph{Q-resolution}~\cite{BuningKarpinskiFlogel95-resolution},
a QBF proof system which is heavily studied and used,
in both theory and practice
(see for example~\cite{BalabanovJiang11-resolutionskolem,GoultiaevaVanGelderBacchus11-uniform,VanGelder12-contributions,JanotaGrigoreMarques-Silva13-preprocessing,JanotaMarquesSilva13-expansions,BalabanovWidlJiang14-qbfcomplexities,BeyersdorffChewJanota14-unification} and the references therein).
When applied to SAT instances
(viewed as instances of QBF where all variables are existentially quantified),
Q-resolution behaves identically to resolution
(a heavily studied propositional proof system), and hence
the known exponential lower bounds on 
resolution proof length~\cite{Haken85-resolution,Ben-SassonWigderson01-narrow}
transfer immediately to Q-resolution.
This observation leaves one with a lingering 
sentiment---which is often expressed by members of the community---that
 there is something left to be said.
After all, Q-resolution is defined on QBF instances,
which are substantially more general than SAT instances;
the observation does not yield any information
about how Q-resolution handles this extra generality,
that is, how it copes with alternation of quantifiers.
Indeed, there is a sharp disconnect
between observing a lower bound for a QBF proof system
via a set of SAT instances, and 
the mentioned treatment of the SAT problem, 
by QBF algorithms, as
a feasibly solvable primitive.
These considerations naturally lead to the question
of whether or not one can formulate and prove a lower bound
which arises from alternation.

{\bf Contributions.}
In this article, we present and study a framework
in which it is possible to present
such alternation-based lower bounds on proof length 
in QBF proof systems.

We define a \emph{proof system ensemble} 
to be an infinite collection of proof systems,
where in each proof system, whether or not a given 
string $\pi$ constitutes a proof of a given formula $\Phi$
can be checked in the polynomial hierarchy
(Definition~\ref{def:proof-system-ensemble}).
A proof system ensemble is considered to have 
\emph{polynomially bounded proofs} (for a language) 
if it contains a proof system
which has polynomially bounded proofs in the usual sense
(Definition~\ref{def:polybounded}).
As a result, it is straightforward to define
proof system ensembles that have succinct proofs
for any set of QBFs with bounded alternation,
such as a set of SAT instances
(and the proof system ensembles studied herein all have this property);
this in turn forces proof length lower bounds,
by nature, 
to arise from a
proof system's inability to cope with quantifier alternation.\footnote{
  Note that there is, a priori, a difference between allowing proof systems
  oracle access to the SAT problem---which 
  would be natural for modelling QBF solvers that 
  treat the SAT problem as feasibly solvable---and allowing oracle access to 
  arbitrary levels of the PH.  
  We focus on the latter for various reasons:
  the proof length lower bounds will arise from alternation;
  we believe that this results in a more robust model; 
  and, this focus causes the proof length lower bounds,
  which are here of primary interest, to be stronger.
}
In terms of complexity classes, 
the question of whether or not there 
exists a polynomially bounded proof system ensemble
for the QBF problem (or any other PSPACE-complete problem)
is equivalent to the question of whether or not PSPACE
is contained in PH, the polynomial hierarchy
(Proposition~\ref{prop:polybounded-iff-in-ph}).
Indeed, the relationship that traditional proof complexity
bears to the NP equals coNP question 
is analogous to the relationship between the present framework
and the PSPACE equals PH question.
(Let us point out that no direct implication is known between
these two open questions, and so, in a certain sense,
progress in one framework may proceed orthogonally to
progress in the other!)

One of our main motivations in pursuing this work
was to gain further insight
into Q-resolution; here, we focus on a slight extension,
\emph{QU-resolution}~\cite{VanGelder12-contributions}, where from existing clauses
one can derive new clauses in two ways:
by a
rule for eliminating literals on universally quantified variables
and
by resolving two clauses on any variable
(in Q-resolution, one can only resolve on existentially quantified variables).
Q-resolution, QU-resolution, and their relatives are typically
defined only for clausal QBFs---QBFs that consist of a quantifier prefix followed by a conjunction of clauses.
We show how to parameterize and lift QU-resolution
to obtain a proof system ensemble
which we call \emph{relaxing QU-res} which is 
in fact defined on arbitrary QBFs (indeed, it is defined on 
what we call \emph{quantified Boolean circuits}),
and not just those in clausal form; 
relaxing QU-res is the main proof system ensemble that we study.
Let us overview how we define it.
\begin{itemize}
\item We define an \emph{axiom} 
of a QBF to be a clause which is, in a certain precise sense,
entailed by the QBF (see Section~\ref{subsect:qu-res}).

\item We then show that, given a QBF $\Phi$ and a partial assignment $a$ to
some of its variables, one can define a QBF $\Phi[a]$
derived naturally from $\Phi$, where the variables on which $a$ is defined
have been instantiated 
(in a certain precise sense; see Section~\ref{subsect:relaxing}).
This QBF $\Phi[a]$ has the key property that
if it is false, then the clause corresponding to $a$ is an axiom
of the QBF $\Phi$ (see Proposition~\ref{prop:derived-clauses} for a precise statement).  
We view the notion of 
inferring clauses from the falsity of QBFs whose variables are
partially instantiated as highly natural; indeed, in the case of SAT,
performing such inferences
is a basis of modern backtracking SAT solvers that perform \emph{clause learning}.

\item 
Recall that each proof system in our proof system ensemble
may use, as an oracle, a level of the PH;
in particular, the QBF problem restricted to a constant
number of alternations may be used as an oracle.
In order to infer clauses from a QBF $\Phi$ 
using the method just described,
we need a way of detecting falsity of QBFs having the form $\Phi[a]$.
But in general, this is difficult; such a QBF $\Phi[a]$ may 
have a high number of alternations, and thus
might not be immediately decidable using an oracle of the described form.
To the end of permitting the falsity detection of QBFs $\Phi[a]$
using such oracles, we define the notion of a
\emph{relaxation} of a QBF.
A relaxation of a QBF $\Phi$ is obtained from $\Phi$ by changing 
the order of the quantifier/variable pairs in the quantifier prefix;
roughly speaking, such a pair $Q v$ may be moved to the
left if $Q$ is the universal quantifier ($\forall$),
and may be moved to the right if $Q$ is the existential quantifier
($\exists$).  (See Section~\ref{subsect:relaxing} for the precise definition.)
A key property of this notion is that if a relaxation of a QBF $\Phi$
is false, then the QBF $\Phi$ is false (Proposition~\ref{prop:relaxation}).

With this notion of relaxation in hand,
we define, for each $k \geq 2$,
the set $H(\Phi, \Pi_k)$ to contain the axioms
that arise from QBFs $\Phi[a]$ having $\Pi_k$-relaxations
(relaxations with a $\Pi_k$ prefix)
that are false.  That is, in this set we collect the axioms 
obtainable
by detecting falsity of QBFs $\Phi[a]$
via the consideration of $\Pi_k$-relaxations.
(Hence, the detection is sound in that it is always correct, 
but it is not complete).
Note that it holds that
$$H(\Phi, \Pi_2) \subseteq H(\Phi, \Pi_3) \subseteq H(\Phi, \Pi_4) \subseteq \cdots.$$


\item This gives us a sequence of versions of QU-resolution:
for each $k$, we obtain a version by defining a proof
to be a sequence of clauses derived from
the axioms $H(\Phi, \Pi_k)$ and the two aforementioned rules
of QU-resolution.
This sequence is the proof system ensemble
\emph{relaxing QU-res}.
Let us remark that each of these versions is 
sound and complete, in a precise sense 
(see Definition~\ref{def:proof-system-ensemble}
and Proposition~\ref{prop:relaxing-qu-res-an-ensemble}).
\end{itemize}


A couple of remarks are in order.
First, note that the empty clause is an axiom in $H(\Phi, \Pi_k)$
whenever $\Phi$ is a false QBF whose quantifier prefix is $\Pi_k$.
Consequently, relaxing QU-res is polynomially bounded
on any set of false QBFs 
\confversion{having bounded alternation.}\fullversion{having bounded alternation (this is discussed in Section~\ref{subsect:relaxing}).}
Let us also note that although here we explicitly lift QU-resolution
to a proof system ensemble, the approach that we take here
can be applied to analogously lift 
any proof system which is based on deriving clauses
from a set of axiom clauses.

Apart from the formulation of the framework,
our main results are as follows.
We prove an exponential separation between the
tree-like and general versions of relaxing QU-res
(Section~\ref{sect:separation-tree-like}),
by exhibiting a set of formulas which have polynomial size
QU-resolution proofs, but which require exponential size proofs
in tree-like relaxing QU-res; this gives
an alternation-based analog of the
known separation between tree-like and general 
resolution~\cite{BEGJ00-relative,Ben-SassonImpagliazzoWigderson04-separation}.
Tree-like QU-resolution proofs can be viewed
as the traces of a natural backtrack-style QBF decision procedure
(this is evident from the viewpoint in Section~\ref{subsect:graph-based-view},
and is also developed explicitly in~\cite[Section 4.3]{Chen14-beyond-qresolution}),
and so this separation formally differentiates the
power of such backtracking and general QU-resolution.
The lower bound of this separation is based on a 
prover-delayer game for tree-like QU-resolution proofs
(Section~\ref{sect:prover-delayer-game}), 
which can be viewed as a generalization of
a known prover-delayer game
for tree-like resolution~\cite{PudlakImpagliazzo00-lowerbound};
note that recently and independently of our work~\cite{Chen14-pc-corr}, 
a game similar to ours was presented
for tree-like Q-resolution~\cite{BeyersdorffChewSreenivasaiah14-game}.
We also prove an exponential lower bound for
relaxing QU-res (Section~\ref{sect:lower-bound-relaxing-qu-res}).

All in all, the ideas and techniques
developed in this work draw upon and interface concepts
from two-player game interaction, proof complexity,
and quantified propositional logic.
We believe that further progress could benefit
from creative input from each of these areas, and
certainly look forward to future research on 
the presented framework.

\fullversion{{\bf Note that some proofs have been deferred to the appendix.}}

\section{Preliminaries}
\label{sect:preliminaries}

For each integer $k$, we 
use $[k]$ to denote the set
that is equal to $\{ 1, \ldots, k \}$
when $k \geq 1$, and that is equal to the empty set $\emptyset$
when $k < 1$.
We use $\nats$ to denote the natural numbers $\{ 0, 1, 2, \ldots \}$.

We use $\dom(f)$ to indicate the domain of a function.
A function $f$ is a 
\emph{restriction} of a function $g$
if $\dom(f) \subseteq \dom(g)$ and,
for each $a \in \dom(f)$, it holds that $g(a) = f(a)$;
when this holds, we also say that $g$ is an 
\emph{extension} of $f$.
When $f$ is a function, we use $f[a \to b]$
to denote the function on domain $\dom(f) \cup \{ a \}$
that maps $a$ to $b$, and otherwise behaves like $f$.
We write $f \res S$ to denote the restriction
of a function $f$ to the set $S$.
We say that two functions $f$ and $g$ \emph{agree}
if for each element $a \in \dom(f) \cap \dom(g)$,
it holds that $f(a) = g(a)$.

When $A$ and $B$ are sets, we use $[A \to B]$
to denote the set of functions from $A$ to $B$.

{\bf Clauses.}
In this article, we employ the following terminology to discuss
clauses.
A \emph{literal} is a propositional variable $v$ or the negation
$\overnot{v}$ thereof.
Two literals are \emph{complementary}
if one is a variable $v$ and the other is $\overnot{v}$;
each is said to be the \emph{complement} of the other.
A \emph{clause} is a disjunction of literals that 
contains, for each variable, at most one literal on the variable.
A clause is sometimes viewed as the set of the literals that
it contains; two clauses are considered equal
if they are equal as sets.
A clause is \emph{empty} if it does not contain any literals.
The variables of a clause are simply the variables that
underlie
the clause's literals, and the set of variables of a clause $\alpha$
is denoted by $\vars(\alpha)$.
When $\alpha$ is a clause, we use $\assign(\alpha)$
to denote the unique propositional assignment $f$
with $\dom(f) = \vars(\alpha)$ such that $\alpha$ 
evaluates to false under $f$.
In the other direction, when $f$
is a propositional assignment, we use
$\clause(f)$ to denote the unique clause $\alpha$
with $\vars(\alpha) = \dom(f)$ that evaluates to false under $f$.
\emph{We will freely and tacitly interchange between
a clause $\alpha$ and its corresponding assignment $\assign(\alpha)$.}
A clause $\gamma$ is 
a \emph{resolvent} of two propositional clauses $\alpha$ and $\beta$
on variable $v$
if there exists a literal $L \in \alpha$ 
such that its complement $M$
is in $\beta$, 
$\gamma = (\alpha \setminus \{ L \}) \cup (\beta \setminus \{ M \})$,
and $v$ is the variable underlying $L$ and $M$.

{\bf Quantified Boolean circuits and formulas.}
We assume basic familiarity with
quantified propositional logic.
A \emph{QBC} (short for \emph{quantified Boolean circuit})
consists of a quantifier prefix
$\prp = Q_1 v_1 \ldots Q_n v_n$,
where each $Q_i$ is a quantifier in $\{ \forall, \exists \}$
and each $v_i$ is a propositional variable;
and, a Boolean circuit $\phi$
built from the constants $0$ and $1$, propositional variables
among $\{ v_1, \ldots, v_n \}$, and
the gates AND ($\wedge$), OR ($\vee$), and NOT ($\neg$).
We refer to the computational problem of deciding whether or not
a QBC is false as the \emph{QBC problem}.
For brevity, we sometimes refer to 
existentially quantified variables as $\exists$-variables,
and
universally quantified variables as $\forall$-variables.
While it is typical to notate a QBC by simply specifying
the prefix $\prp$ immediately followed by the circuit $\phi$,
we will typically separate these two parts by a colon
for the sake of readability,
using for example $\prp:\phi$.
We assume that each 
quantifier prefix does not contain repeated variables.
When $\Phi = \prp:\phi$ is a QBC, 
by a \emph{partial assignment of $\Phi$},
we refer to a propositional assignment
$f: S \to \{ 0, 1 \}$
defined on a subset $S$ of the variables appearing in $\prp$.
A \emph{QBF} is a QBC $\prp:\phi$ where $\phi$ is 
a Boolean formula.  A \emph{clausal QBF}
is a QBF $\prp:\phi$ where
$\phi$ is the conjunction of clauses.

{\bf Quantifier prefixes.}
Let $i \geq 1$.
A quantifier prefix $\prp = Q_1 v_1 \ldots Q_n v_n$
is $\Pi_i$ if $Q_1 \ldots Q_n$, viewed as a string over the alphabet
$\{ \forall, \exists \}$, is contained in the language denoted
by the regular expression 
$\forall^* \exists^* \forall^* \exists^* \ldots$,
which contains $i$ starred quantifiers,  
beginning with $\forall^*$ and alternating;
$\Sigma_i$ is defined similarly, but with respect to
the regular expression
$\exists^* \forall^* \exists^* \forall^* \ldots$.

The following notation is relative to a quantifier prefix
$\prp = Q_1 v_1 \ldots Q_n v_n$; when we use it,
the prefix will be clear from context.
We write $v_i \preceq v_j$ if $i \leq j$
or if $j < i$ and $Q_j = Q_{j+1} = \cdots = Q_i$.
We extend this binary relation (and others) to sets in 
the following natural way:
when $U$ and $V$ are sets of variables,
we write $U \preceq V$ if for each $u \in U$ and each $v \in V$,
it holds that $u \preceq v$.  We also write,
for example, that $U \preceq v$ for a single variable $v$
when $U \preceq \{ v \}$.
We write $v_i \equiv v_j$ if $v_i \preceq v_j$ and 
$v_j \preceq v_i$.  
It is straightforward to verify that $\equiv$ is
an equivalence relation; we refer to each equivalence class
of $\equiv$ as a \emph{quantifier block}. 
We write $v_i \precneq v_j$ if $v_i \preceq v_j$ and
$v_i \not\equiv v_j$.
When $S$ is a set of variables, we use
$\last(S)$ to denote the variable of $S$ appearing last
in the quantifier prefix, that is, the variable $v_m$,
where $m = \max \{ i ~|~ v_i \in S \}$.
Typically, when we use the function $\last(S)$,
it is in conjunction with the just-defined binary relations,
and hence what is most relevant will be the 
relative location of the quantifier block 
of $\last(S)$.

{\bf Strategies.}
Let $\Phi = \prp:\phi$ be a QBC;
let $X$ denote the $\exists$-variables of $\Phi$,
and let $Y$ denote the $\forall$-variables of $\Phi$.
When $x \in X$, define $Y_{<x}$ to be the set
of variables $\{ y \in Y ~|~ y \precneq x \}$;
dually, when $y \in Y$, define $X_{<y}$
to be the set
of variables $\{ x \in X ~|~ x \precneq y \}$.

An \emph{$\exists$-strategy} is a sequence of mappings
$ \sigma = (\sigma_x)_{x \in X}$ where each $\sigma_x$
is a mapping from $[Y_{<x} \to \{ 0, 1 \}]$ to $\{ 0, 1 \}$.
When $\tau: Y \to \{ 0, 1 \}$ is an assignment to the
universally quantified variables,
we use $\langle \sigma, \tau \rangle$ to denote the
assignment $f$ defined by $f(y) = \tau(y)$ for each $y \in Y$
and $f(x) = \sigma_x(\tau \res Y_{<x})$ for each $x \in X$.
We say that $(\sigma_x)_{x \in X}$ 
is a \emph{winning $\exists$-strategy} if for every assignment
$\tau: Y \to \{ 0, 1 \}$, it holds that 
the assignment $\langle \sigma, \tau \rangle$
satisfies $\phi$.
A \emph{model} of $\Phi$ is defined to be a
winning $\exists$-strategy of $\Phi$.

Dually, we define
a \emph{$\forall$-strategy} to be a sequence of mappings
$ \tau = (\tau_y)_{y \in Y}$ where each $\tau_y$
is a mapping from $[X_{<y} \to \{ 0, 1 \}]$ to $\{ 0, 1 \}$.
When $\sigma: X \to \{ 0, 1 \}$ is an assignment to the
existentially quantified variables,
we use $\langle \tau, \sigma \rangle$ to denote the
assignment $f$ defined by $f(x) = \sigma(x)$ for each $x \in X$
and $f(y) = \tau_y(\sigma \res X_{<y})$ for each $y \in Y$.
We say that $(\sigma_y)_{y \in Y}$ 
is a \emph{winning $\forall$-strategy} if for every assignment
$\sigma: X \to \{ 0, 1 \}$, it holds that 
the assignment $\langle \tau, \sigma \rangle$
falsifies $\phi$.

The following are well-known facts that we will treat as basic.

\begin{prop}
Let $\Phi$ be a QBC.
\begin{itemized}

\item There exists a winning $\exists$-strategy for $\Phi$
(that is, a model of $\Phi$)
if and only if $\Phi$ is true.

\item There exists a winning $\forall$-strategy for $\Phi$
if and only if $\Phi$ is false.

\end{itemized}
\end{prop}

\section{Proof system ensembles}

In this section, we formalize the notion of  
\emph{proof system ensemble} and present some basic
associated notions.

For each $m \geq 1$,
fix $S(m)$ to be 
the QBC problem restricted to QBCs
having a $\Sigma_m$ prefix, which is
a $\Sigmap_m$-complete problem;
for $m = 0$, fix $S(m)$ to be a polynomial-time decidable problem.

Let $O$ be a language;
when discussing an algorithm $A$ that makes oracle calls,
we use $A^{O}$ to denote the instantiation of $A$
where oracle calls are answered according to $O$.

\begin{definition}
\label{def:proof-system-ensemble}
A \emph{proof system ensemble $(A,r)$
for a language $L$}
consists of 
an algorithm $A$ which 
may make oracle calls and
receives inputs of the form 
$(k,(x,\pi))$ where $k \in \nats$ and $x$
and $\pi$ are strings;
and, a computable function $r: \nats \to \nats$
such that: 
\begin{itemized}

\item For each $k \in \nats$,
there exists a polynomial $p_k$
such that (for each pair $(x, \pi)$)
the algorithm $A^{S(r(k))}$
halts on an input $(k, (x, \pi))$
within time $p_k(|(x, \pi)|)$.

\item For each $k \in \nats$,
when $L_k$ is set to 
$\{ (x,\pi) ~|~ (k,(x,\pi)) \textup{ is accepted by $A^{S(r(k))}$} \}$,
it holds that the language
$\{ x ~|~ \exists \pi \textup{ such that $(x,\pi) \in L_k$} \}$
is equal to $L$.

\end{itemized}
\end{definition}

Let us provide an intuitive explanation of 
Definition~\ref{def:proof-system-ensemble}.
For each fixed value of $k$, the algorithm $A$ provides
a proof system for the language $L$;
on inputs of the form $(k, (x, \pi))$,
the algorithm is provided oracle access to $S(r(k))$,
and needs to accept or reject within polynomial time
(in $|(x,\pi)|$).
Acceptance indicates that $\pi$ is judged to be a proof 
that $x \in L$.
The second condition in the definition
states that each such proof system 
is sound and complete,
that is, for each fixed $k$,
an arbitrary string $x$ is in $L$
iff there exists a string $\pi$
such that $(k,(x,\pi))$ is accepted by $A$.

We use the following terminology to present lower bounds
on proof size in proof system ensembles.

\begin{definition}
\label{def:exp-proofs}
Let $Z$ be a set of functions from $\nats$ to $\nats$.
A proof system ensemble $(A, r)$ 
\emph{requires proofs of size $Z$}
on a sequence $\{ \Phi_1, \Phi_2, \ldots \}$ of instances
if for each $k$, there exists $z \in Z$
where
(for all $n \geq 1$ and all strings $\pi$)
it holds that $(k, (\Phi_n, \pi)) \in L_k$ implies
$|\pi| \geq z(n)$.
Here, $|\pi|$ denotes the size of $\pi$.
We also apply this terminology to other measures
defined on proofs.
\end{definition}

We say that a function $f$ mapping strings to strings is
a \emph{polynomial-length function} 
if there exists a polynomial $q$ such that,
for each string $x$, it holds that $|f(x)| \leq q(|x|)$.

\begin{definition}
\label{def:polybounded}
A proof system ensemble $(A, r)$ is
\emph{polynomially bounded} on a language $L$
if there exists $k \in \nats$
and there exists a polynomial-length function $f$ 
(mapping strings to strings)
such that
the following holds: 
if $x \in L$, then it holds that $(x, f(x)) \in L_k$,
where $L_k$ is defined as in
Definition~\ref{def:proof-system-ensemble}.
\end{definition}

\begin{prop}
\label{prop:polybounded-iff-in-ph}
There exists a polynomially bounded proof system ensemble
for a language $L$ if and only if $L$ is in the polynomial hierarchy.
\end{prop}

\fullversion{
We next define notions of simulation between proof systems.

\begin{definition}
\label{def:simulates}
Let $(A, r)$ and $(A', r')$ be proof system ensembles
for a language $L$.

We say that $(A', r')$ \emph{simulates} $(A, r)$ 
if there exists a function $f: \nats \to \nats$
and a 
sequence of polynomial length functions 
$(g_k)_{k \in \nats}$ from strings to strings
such that, 
for each $k \in \nats$ and each $(x, \pi) \in L_k$,
it holds that $(x, g_k(\pi)) \in L'_{f(k)}$.
Here, $L_k$ and $L'_k$ are defined as
in Definition~\ref{def:proof-system-ensemble},
for $(A, r)$ and $(A', r')$, respectively.

We say that $(A', r')$ \emph{effectively simulates} $(A, r)$
if, in addition, the function $f$ is computable
and there is an algorithm that,
for each $k \in \nats$, computes $g_k(x)$ from $x$ 
within time $p_k(|x|)$,
where $p_k$ is a polynomial.
\end{definition}
}

\fullversion{
Under a mild assumption on proof system ensembles,
namely that (intuitively) they increase in strength 
as the parameter $k$ increases, it can be proved that,
when $(A,r)$ and $(A',r')$ are proof system ensembles
such that $(A', r')$ simulates $(A,r)$
and such that $(A,r)$ is polynomially bounded,
it holds that $(A',r')$ is polynomially bounded.
We will formalize and discuss this in the full version of the article.
}

\fullversion{
Let us remark that variations on 
Definitions~\ref{def:proof-system-ensemble}
and~\ref{def:simulates}
are certainly possible.
For example, one could require that the bounding
polynomials $(p_k)$ in 
Definition~\ref{def:proof-system-ensemble}
be computable, as a function of $k$.
Perhaps more interestingly, 
observe that no assumption is placed
on how these polynomials $(p_k)$ 
behave in aggregate;
one could, for instance, require that their degrees are all
bounded above by a constant, obtaining a definition
reminiscent of that of fixed-parameter tractability.
A similar comment can be offered for the polynomials 
associated to the functions $(g_k)$
from Definition~\ref{def:simulates}.
}

\section{Relaxing QU-resolution}

\subsection{QU-resolution}
\label{subsect:qu-res}

Let $\Phi = \prp:\phi$ be a QBC. 
We define an \emph{axiom set of $\Phi$} 
to be a set $H$ of clauses on variables of $\prp$
such that, for each $C \in H$,
$C$ is an \emph{axiom} of $\Phi$ in the following sense:
each model of $\prp:\phi$
is a model of $\prp:C$.
Let us give  examples.
First, if the QBC $\Phi$ is false,
then the empty clause is an axiom of $\Phi$.
Second, if $C$ is any clause which is entailed by $\phi$,
then $C$ is an axiom of $\Phi$. 
A case of this
is when $a$ is an assignment to all variables of $\Phi$
that falsifies $\phi$; then, $\clause(a)$ is entailed by $\phi$
and is an axiom of $\Phi$.

Relative to a QBC $\Phi = \prp:\phi$,
we say that
a clause $C$ is obtainable from a second clause $D$
by \emph{$\forall$-elimination} if there exists a literal $L \in D$
such that $C = D \setminus \{ L \}$
and the variable $y$ underlying $L$ is a $\forall$-variable and 
has $\vars(C) \preceq y$.

With these notions, we define QU-resolution 
for quantified Boolean circuits in the following way.

\begin{definition}
\label{def:qu-resolution-proof}
A \emph{QU-resolution proof} of a QBC $\Phi = \prp:\phi$
from an axiom set $H$ (of $\Phi$)
is a finite sequence of clauses
where each clause is either in $H$,
is obtainable from a previous clause by $\forall$-elimination,
or is obtainable from two previous clauses
as a resolvent; 
in the last two cases, we assume that the clause is annotated
with the previous clause(s) from which it is derived
(this is to provide a clean correspondence between
proofs and certain graphs to be defined, 
see Section~\ref{subsect:graph-based-view}).
The \emph{size} of such a proof is defined as the number of clauses.
Such a proof is said to be a \emph{falsity proof} if it
ends with the empty clause.
\end{definition}

\fullversion{
Note that in the case that $\Phi$ is a clausal QBF,
when $H$ is the set of clauses appearing in $\Phi$,
Definition~\ref{def:qu-resolution-proof}
essentially coincides with usual definitions of QU-resolution
in the literature 
(see for example~\cite{JanotaGrigoreMarques-Silva13-preprocessing}).
The only difference is that here, 
applying $\forall$-elimination eliminates just one 
universally quantified variable of a clause,
whereas many authors speak of \emph{$\forall$-reduction},
which (when applied to a clause) 
eliminates each universally quantified variable
that come after all existentially quantified variables.
One can simulate an instance of $\forall$-reduction by
applying $\forall$-elimination repeatedly.
}

It is a folklore and readily verified fact that when one has a 
clausal QBF $\Phi = \prp:\phi$
with clause set $H$, and $C$ appears in a QU-resolution
proof of $\Phi$ from $H$, then
any model of $\Phi$ is a model of $\prp:C$.
From this fact and the definition of axiom set,
we immediately obtain the following proposition.

\begin{prop}
\label{prop:soundness-qu-res}
Let $C$ be a clause appearing
in a QU-resolution proof of a QBC $\Phi = \prp : \phi$
from axiom set $H$.
Each model of $\prp : \phi$ is a model of $\prp : C$.
Consequently, if $C$ is the empty clause, then the QBC $\Phi$
is false.
\end{prop}

\subsection{Relaxing}
\label{subsect:relaxing}

In order to define a proof system ensemble based on QU-resolution
proofs, we now describe how to obtain a sequence of axiom sets
for a given QBC.  We start by exhibiting a way to 
infer that a partial assignment is an axiom of a QBC.

Let $a$ be a partial assignment
of a QBC  $\Phi = \prp : \phi$.
Define $\prp[a]$ to be the quantifier prefix which is equal to $\prp$
but where the variables in $\dom(a)$ and their corresponding
quantifiers are removed, and where each 
quantifier 
of a variable $v$ with $v \precneq \last(a)$
is changed
(if necessary) to an existential quantifier.  
\fullversion{As examples, when 
$\prp = \forall y_1 \exists x_1 \exists x_2 \forall y_2 \forall y_3 \exists x_3$,
if $a$ is an assignment with $\dom(a) = \{ x_1, y_3 \}$,
it holds that $\prp[a] = \exists y_1 \exists x_2 \forall y_2 \exists x_3$;
if $a$ is an assignment with $\dom(a) = \{ x_1, x_2 \}$,
it holds that $\prp[a] = \exists y_1 \forall y_2 \forall y_3 \exists x_3$.}
Define $\phi[a]$ 
to be the circuit obtained from $\phi$ by replacing
each variable $v \in \dom(a)$ with the constant $a(v)$.
Define $\Phi[a]$ to be $\prp[a] : \phi[a]$.

\begin{prop}
\label{prop:derived-clauses}
Assume that 
$a$ is a partial assignment
of a QBC  $\Phi = \prp : \phi$
such that
$\Phi[a]$ is false.
Then $\clause(a)$ is an axiom of $\Phi$, that is,
each model of $\prp : \phi$ is a model of $\prp : \clause(a)$.
\end{prop}

We believe that Proposition~\ref{prop:derived-clauses}
provides a natural way to derive axioms from a QBC.
Consider the case where $\Phi$ is a SAT instance,
that is, $\prp$ is purely existential.
In this case, if $a$ is a partial assignment
such that $\Phi[a]$ is false,
then $\clause(a)$ is an axiom of $\Phi$.
Indeed, in this case $\Phi[a]$ is simply the QBC instance obtained
by instantiating variables according to $a$, and then
removing the instantiated variables from the quantifier prefix.
Note that,
in the context of backtrack search for SAT, it is typical
that, when some variables have been set according
to a partial assignment $a$, 
a solver attempts to detect falsity of $\Phi[a]$
by heuristics such as unit propagations and generalizations thereof.

In the case of general QBCs, it is natural to ask,
when one has a partial assignment $a$
and then instantiates its variables in $\phi$ to obtain $\phi[a]$, 
under what conditions $\clause(a)$ can be inferred as an axiom.
Proposition~\ref{prop:derived-clauses} provides an
answer to this question; let us explain intuitively why the quantifier prefix
is adjusted to $\prp[a]$.
Consider the case where the first quantifier block of
$\prp$ is existential and $a$ is a partial assignment
to variables from this first block; then $\prp[a]$
is simply $\prp$ but with the variables of $a$ removed,
and so this case of the proposition generalizes the
purely existential case just discussed.
In the case where $a$ is arbitrary, $\prp[a]$
can be viewed as the prefix where the lowest number of quantifiers
have been changed from universal to existential such that
the first quantifier block is existential, and
all variables of $a$ fall into this first block.

\fullversion{
Proposition~\ref{prop:derived-clauses}
can be proved in the following way.
Fix a model $\sigma = (\sigma_x)_{x \in X}$ of $\prp:\phi$;
here, $X$ denotes the $\exists$-variables in $\prp$.
Suppose (for a contradiction) that $\tau$ is an assignment to the
$\forall$-variables of $\prp:\phi$
such that the assignment $f = \langle \sigma, \tau \rangle$
falsifies $\clause(a)$, or equivalently,
$f$ extends the assignment $a$.
Then, we define a winning $\exists$-strategy $\sigma'$ 
for $\Phi[a]$
as follows.
Define $\sigma'_x$ to be the function obtained from $\sigma_x$
after fixing each $\forall$-variable
$y \in \dom(a) \cup \{ v ~|~ v \precneq \last(a) \}$
to $\tau(y)$;
and, for each $\forall$-variable $y$ with $y \precneq \last(a)$
(that is, for each $\forall$-variable in $\prp$ that is changed
to an $\exists$-variable in $\prp[a]$),
 define $\sigma'_y$ to be $\tau(y)$.
}

Prima facie, Proposition~\ref{prop:derived-clauses}
may appear to be of limited utility; even if one
has oracle access to a level of the polynomial hierarchy,
it may be that many partial assignments $a$ give rise to
a quantifier prefix $\prp[a]$ which has too many alternations
to be resolved by the oracle.
In order to expand the class of axioms derivable by this proposition
(relative to such an oracle), we introduce now the notion of 
a \emph{relaxation} of a QBC.

A \emph{relaxation} of a quantifier prefix 
$\prp = Q_1 v_1 \ldots Q_n v_n$
is a quantifier prefix which has the form
$\prpp = Q_{\pi(1)} v_{\pi(1)} \ldots Q_{\pi(n)} v_{\pi(n)}$
where $\pi: [n] \to [n]$ is a permutation
and where, for each $\forall$-variable $y$
and for each $\exists$-variable $x$,
it holds that $y \preceq x$ implies $y \preceq' x$;
here, $\preceq$ and $\preceq'$ denote the
binary relations of $\prp$ and $\prp'$, respectively.
As an example, consider the quantifier prefix
 $\prp = \exists x_1 \exists x_2 \forall y \forall y' \exists x_3$;
relaxations thereof include 
$\forall y \forall y' \exists x_1 \exists x_2 \exists x_3$,
$\exists x_1 \forall y' \exists x_2 \forall y \exists x_3$,
and
$\forall y' \exists x_2 \forall y \exists x_1 \exists x_3$.
A \emph{relaxation of a QBC $\prp:\phi$} 
is a QBC of the form $\prpp:\phi$
where $\prpp$ is a relaxation of $\prp$;
such a QBC is said to be a 
\emph{$\Pi_i$-relaxation} if $\prpp$ is $\Pi_i$.

The following is straightforward to verify.

\begin{prop}
\label{prop:relaxation}
If a relaxation of a QBC $\Phi$ is false, 
then the QBC $\Phi$ is false.
\end{prop}

Note that for any quantifier prefix, a relaxation
may be obtained by simply placing the
universal quantifiers and their variables first, followed
by the existential quantifiers and their variables.
Hence, in this sense, each QBC has a canonical $\Pi_2$-relaxation,
and
in the sequel, we focus the discussion
on relaxations that are $\Pi_k$-relaxations
for values of $k$ greater than or equal to $2$. 

Let $\Phi$ be a QBC;
for $k \geq 2$, 
we define $H(\Phi, \Pi_k)$ to be the set that contains
a clause $C$ if there exists a $\Pi_k$-relaxation of
$\Phi[\assign(C)]$ that is false.
The following fact follows immediately from
Propositions~\ref{prop:derived-clauses}
and~\ref{prop:relaxation}.

\begin{prop}
\label{prop:relaxations-give-axiom-set}
When $\Phi$ is a QBC and $k \geq 2$,
it holds that $H(\Phi, \Pi_k)$ is an axiom set of $\Phi$.
\end{prop}

\fullversion{
Note that when $\Phi = \prp:\phi$ 
is a clausal QBF, $C$ is a clause in $\phi$,
and $a = \assign(C)$,
it holds that $\phi[a]$ is unsatisfiable;
consequently,
for any quantifier prefix $\prpp$ on the variables of $\phi[a]$,
it holds that $\prpp:\phi[a]$ is false,
and thus $C \in H(\Phi, \Pi_2)$.
Hence, the set $H(\Phi, \Pi_2)$ contains each clause of $\phi$.
}

\begin{definition}
\emph{Relaxing QU-res} 
is defined as the pair $(A, r)$
where $r$ is defined by $r(k) = k+3$
and $A$ is an algorithm defined to accept an input
$(k,(\Phi,\pi))$
if $\Phi$ is a QBC and 
$\pi$ is a QU-resolution falsity proof of $\Phi$
from axioms
in
$H(\Phi, \Pi_{k+2})$.
In particular, the algorithm $A$ 
examines each clause in $\pi$ in order;
when a clause $C$ is not derived from previous ones
by resolution or by $\forall$-elimination,
membership of $C$ in $H(\Phi, \Pi_{k+2})$
is checked by the $\Sigma_{k+3}$ oracle.
(Such an oracle can nondeterministically guess
a $\Pi_{k+2}$-relaxation and then check this relaxation for falsity.)
\end{definition}

\begin{prop}
\label{prop:relaxing-qu-res-an-ensemble}
\emph{Relaxing QU-res} is a proof system ensemble
for the language of false QBCs.
\end{prop}

\fullversion{
Note that for any set $\fancyf$ of false QBCs
having bounded alternation, 
it holds that \emph{relaxing QU-res}
is polynomially bounded on $\fancyf$.
Why?  Let $k$ be a value such that each QBC in $\fancyf$
is $\Pi_{k+2}$.
For each QBC $\Phi \in \fancyf$,
we have that the empty clause is in 
$H(\Phi, \Pi_{k+2})$, since
$\Phi$ itself is a false $\Pi_{k+2}$-relaxation of $\Phi$.
Hence, for each such QBC $\Phi$,
the algorithm $A$ of relaxing QU-res
accepts $(k, (\Phi, \emptyset))$,
where here $\emptyset$ denotes the 
proof consisting just of the empty clause.
}

Let us now introduce some notions which will be used
in our study of \emph{tree-like relaxing QU-res} (defined below).
Let $f$ and $g$ be partial assignments of a QBC $\Phi$.
We say that $g$ is a \emph{semicompletion} of $f$
if $g$ is an extension of $f$ such that
for each universally quantified variable $y$ with 
$\dom(f) \preceq y$ and $y \notin \dom(f)$,
it holds that
$\dom(g) \preceq y$ and $y \notin \dom(g)$.
A set $H$ of partial assignments of $\Phi$ is
\emph{semicompletion-closed}
if, whenever $f \in H$ and $g$ is a semicompletion of $f$,
it holds that $g \in H$.

\fullversion{
\begin{lemma}
\label{lemma:semicompletion-closure}
For each QBC $\Phi$ and for each $m \geq 2$, 
the set of assignments
$H(\Phi, \Pi_m)$ is semicompletion-closed.
\end{lemma}

}

\subsection{A graph-based view}
\label{subsect:graph-based-view}

When $\pi = C_1, \ldots, C_n$ is a QU-resolution proof of
a QBC $\prp : \phi$ from axioms $H$,
define $G(\pi)$ to be the directed acyclic graph where
there is a vertex for each clause occurrence $C_i$, 
which vertex has label
$\assign(C_i)$;
and, where 
(for all pairs of clauses $C_i, C_j$)
there is a directed edge from the vertex of $C_j$
to the vertex of $C_i$ if $C_j$ is derived from $C_i$.

\begin{prop} \rm
\label{prop:proof-graph}
Let $\pi$ be a QU-resolution proof of a QBC $\prp : \phi$ from axioms $H$.
The directed acyclic graph $G(\pi)$ has the following properties:

\begin{enumerate}
\setlength{\itemsep}{0pt}
  \setlength{\parskip}{0pt}
  \setlength{\parsep}{0pt}

\item[$(\alpha)$] If a node with label $a$ has no out-edges, 
then $\clause(a)$ is an element of $H$.

\item[$(\beta)$] If a node with label $a$ has $1$ out-edge
to a node with label $a'$,
then $a'$ is an extension of $a$ with 
$\dom(a') = \dom(a) \cup \{ y \}$ where $y$ is
a universally quantified variable with $\dom(a) \preceq y$.

\item[$(\gamma)$]
If a node with label $a$ has $2$ out-edges
to nodes with labels $a_1$ and $a_2$,
then there exists a variable $v$
such that $a_1$ and $a_2$ are defined on $v$ and $a_1(v) \neq a_2(v)$;
$(\dom(a_1) \cup \dom(a_2)) \setminus \{ v \} = \dom(a)$;
$a$ and $a_1$ are equal on the variables where they are both defined;
and,
$a$ and $a_2$ are equal on the variables where they are both defined.

\end{enumerate}

Moreover, a labelled graph with these three properties naturally
induces
a QU-resolution proof: for each node, let $a$ be its label, and
associate to it $\clause(a)$. $\Box$
\end{prop}

\begin{definition}
We say that a QU-resolution proof $\pi$
is \emph{tree-like} if the graph $G(\pi)$ is a tree.
We define \emph{tree-like relaxing QU-res}
to be the proof system ensemble $(A', r)$
described as follows.
Let $(A, r)$ denote relaxing QU-res.
Then, the algorithm $A'$ accepts an input
$(k,(x,\pi))$ if $A$ accepts it and $\pi$ is tree-like.
\end{definition}

%% file: technical.tex
\section{A prover-delayer game for \emph{tree-like relaxing QU-res}}
\label{sect:prover-delayer-game}

In this section, we present a game that can be used to 
exhibit lower bounds on the size of tree-like QU-resolution proofs;
this game can be viewed as a generalization of a game
for studying tree-like resolution,
which game was presented by 
Pudl{\'{a}}k and Impagliazzo~\cite{PudlakImpagliazzo00-lowerbound}.

We first give an intuitive description of the game.
Note, however, that this description is meant only to be suggestive.
For a precise description, we urge the reader to consult the formal definition, which follows 
\confversion{(Definition~\ref{def:delayer-strategy}).}
\fullversion{(Definition~\ref{def:delayer-strategy}); 
in this formal definition,
the game is formulated in a positional fashion:
a state of the game is formalized as a partial assignment.}

Relative to a QBC $\Phi$ and a set $H$ of axioms,
the game is played between two players, \emph{Prover}
and \emph{Delayer}, which maintain a partial assignment.
Prover's goal is to reach a partial assignment in $H$,
while Delayer tries to slow down Prover, scoring points in the process.
Prover starts by announcing the empty assignment, 
and Delayer responds with a semicompletion thereof.
After this, the play proceeds in a sequence of rounds.
In each round, Prover may perform one of three actions to 
the current assignment $f$:
select a restriction of $f$;
assign a value to a $\forall$-variable $y \notin \dom(f)$
having $\dom(f) \preceq y$; or,
select a variable $v \notin \dom(f)$.
In the first two cases, Delayer responds with a semicompletion
of the resulting assignment.
In the third case, Delayer may 
give a choice to the Prover.  When a choice is given,
the Prover sets the value of $v$, and Delayer may elect to claim
a point which is then associated with $v$.
When no choice is given, Delayer sets the value of $v$.
After $v$ is set, Delayer responds (as in the first two cases)
with a semicompletion of the resulting assignment.
Delayer is said to have a $p$-point strategy if,
he has a strategy where, by the time that Prover
achieves her goal, 
there are $p$ variables on which the final assignment is defined 
such that Delayer has claimed points on these variables.  \fullversion{\\}
In what follows, we assume $p \geq 1$.

\begin{definition}
\label{def:delayer-strategy}
Let $\Phi$ be a QBC.
Relative to a set $H$ of axioms,
a \emph{$p$-point delayer strategy}
consists of a set $F$ of partial assignments of $\Phi$
and a function $s: F \to \nats$
called the \emph{score function}
such that the following properties hold:
\begin{itemized}

\item 
\emph{(semicompletion-of-empty)}
There exists a semicompletion $g \in F$
of the empty assignment such that $s(g) = 0$.

\item 
\emph{(all-points)}
If $f \in F \cap H$, then $s(f) \geq p$.

\item 
\emph{(monotonicity)} 
If $g \in F$, then each restriction of $g$
has a semicompletion $f \in F$ such that $s(f) \leq s(g)$.

\item
\emph{($\forall$-branching)}
If $f \in F$ and $y \notin \dom(f)$ is
a universally quantified variable with $\dom(f) \preceq y$,
then, for each $b \in \{ 0, 1 \}$, 
the assignment $f[y \to b]$ has a semicompletion $g \in F$
with $s(g) = s(f)$.

\item
\emph{(double-branching)}
If $f \in F$ and $v \notin \dom(f)$,
there exists a value $b \in \{ 0, 1 \}$ 
such that
$f[v \to b]$ has a semicompletion $g \in F$ 
where 
(1)  $s(g) \leq s(f) + 1$
and 
(2) if $s(g) = s(f) + 1$,
the assignment $f[v \to \neg b]$ has a semicompletion $g' \in F$
with $s(g') \leq s(f) + 1$.
\end{itemized}
\end{definition}

\begin{theorem}
\label{thm:delayer-many-leaves}
Assume
that there exists a $p$-point delayer strategy for a 
QBC $\Phi$ with respect to a semicompletion-closed axiom set $H$, 
and 
that $\pi$ is a tree-like QU-resolution proof 
ending with the empty clause,
from axioms $H$.
Then, the tree $G(\pi)$ has at least $2^p$ leaves.
\end{theorem}

\fullversion{
In order to prove Theorem~\ref{thm:delayer-many-leaves},
we introduce the following lemma.

\begin{lemma}
\label{lemma:child}
Assume
that there exists a $p$-point delayer strategy for a 
QBC $\Phi$ with respect to a semicompletion-closed axiom set $H$, 
and that $\pi$ is a tree-like QU-resolution proof
from axioms $H$.
Let $u$ be a node of $G(\pi)$ with label $a$.
If $a$ has a semicompletion $f \in F$
with $s(f) < p$, 
then $u$ has a child $v'$ 
with label $a'$ such that $a'$ has a semicompletion
$g' \in F$ with $s(g') \leq s(f) + 1$;
moreover, when $g' \in F$ has $s(g') = s(f) + 1$,
the node $u$ has a second child $v''$ 
whose label $a''$ has a semicompletion $g''$
with $s(g'') \leq s(f) + 1$.
\end{lemma}

\begin{proof}
Since $s(f) < p$, by the \emph{(all-points)} condition,
we have that $f \notin H$.  Since $H$ is assumed to be
semicompletion-closed, we have that $a \notin H$,
and hence that the node $u$ is not a leaf
of the tree $G(\pi)$.

If the node $u$ is of type $(\beta)$
from Proposition~\ref{prop:proof-graph},
then let $a'$ be the label of the child $v'$ of $u$;
$a'$ has the form $a[y \to b]$.
By the ($\forall$-branching) condition, 
the assignment $f[y \to a'(y)]$ 
has a semicompletion $f'$ with $s(f') \leq s(f)$.
We have that $f'$ is a semicompletion of $a'$, giving the lemma.

If the node $u$ is of type $(\gamma)$
from Proposition~\ref{prop:proof-graph},
let $x$ denote the variable described in the
proposition statement.
We consider two cases.
First, if $x \in \dom(f)$, then pick the child of $u$
with label $a'$ having $a'(x) = f(x)$.
By \emph{(monotonicity)}, the restriction
$f \res \dom(a')$ has a semicompletion $g' \in F$ 
such that $s(g') \leq s(f)$, giving the lemma.
When $x \notin \dom(f)$, we argue as follows.
By the \emph{(double-branching)} condition,
there exists a value $b \in \{ 0, 1 \}$
such that $f[x \to b]$ has a semicompletion $f'$
satisfying the properties (1) and (2) 
given in Definition~\ref{def:delayer-strategy};
in particular, we have $s(f') \leq s(f) + 1$.
Let $v'$ be the child of $u$ whose label $a'$ has $a'(x) = b$.
The assignment $a'$ restricts $a[x \to b]$ 
which restricts $f[x \to b]$, so $f'$ extends $a'$.
By the \emph{(monotonicity)} condition,
the restriction $f' \res \dom(a')$ has a semicompletion
$g'$ with $s(g') \leq s(f')$.
If $s(g') \leq s(f)$, the lemma is proved.
Otherwise, we have that $s(g') = s(f) + 1$,
and by property (2) of \emph{(double-branching)},
the assignment $f[x \to \neg b]$ has 
a semicompletion $f''$ with $s(f'') \leq s(f) + 1$.
Let $v''$ be the child of $u$ whose label $a''$ has $a''(x) = \neg b$.
We have that $a''$ restricts $a[x \to \neg b]$ 
which restricts $f[x \to \neg b]$;
by \emph{(monotonicity)}, the restriction
$f'' \res \dom(a'')$ has a semicompletion $g'' \in F$
such that $s(g'') \leq s(f'')$, giving the lemma.
\end{proof}

\begin{proof} (Theorem~\ref{thm:delayer-many-leaves})
We refer to a semicompletion of the label of a node 
simply as a semicompletion of the node.
We prove by induction on $i = 0, \ldots, p$ 
that, 
for any node $v$ with semicompletion $f$
having $s(f) = p-i$,
the node $v$ has $2^i$ leaf descendents.
This suffices by the property 
\emph{(semicompletion-of-empty)}
of Definition~\ref{def:delayer-strategy}.

The claim is obvious for $i = 0$, 
so suppose that it is true for $i < p$;
we will prove it true for $i+1$.
We have, by assumption, a semicompletion $f$ of $v$
with $s(f) = p-(i+1) = p-i-1$.
Repeatedly invoke Lemma~\ref{lemma:child}
to obtain a path from $v$ to a leaf
where each vertex has a semicompletion associated with it.
Notice that, in walking along this path starting from $v$
and looking at the semicompletions,
whenever the score increases, it increases by at most $1$.
Since any semicompletion of a leaf must have
score $p$ or higher 
(by the reasoning 
at the beginning of the proof of Lemma~\ref{lemma:child}),
there must be some descendent $u$ of $v$
having semicompletion with score $p-i-1$
such that the the child $u_1$ of $u$ provided 
by Lemma~\ref{lemma:child} has semicompletion with score $p-i$.
By that lemma, the other child $u_2$ of $u$
has semicompletion with score less than or equal to $p-i$.
By repeatedly invoking Lemma~\ref{lemma:child}
to obtain a path from $u_2$ to a leaf,
one finds a descendent $v_2$ of $u_2$ 
having a semicompletion with score $p-i$.
By induction, there are at least $2^{i}$ leaves
below $u_1$, and at least $2^{i}$ leaves
below $v_2$ (and hence below $u_2$).
Therefore, there are at least $2^{i+1}$ leaves below $u$,
and hence below $v$.
\end{proof}
}


\section{Separation of the tree-like and general versions of \emph{relaxing QU-res}}
\label{sect:separation-tree-like}

The family of sentences to be studied in this section
is defined as follows.
For each $i \in \{ 0 \} \cup [n]$, define 
$X_i$ to be 
the variable set $\{ x_{i,j,k} ~|~ j, k \in \{ 0, 1 \} \}$,
and for each $i \in [n]$,
define $X'_i$ analogously to be 
the variable set
$\{ x'_{i,j,k} ~|~ j, k \in \{ 0, 1 \} \}$.
Define $\prp_n$ to be the prefix
$\exists X_0 
\exists X'_1 \forall y_1 \exists X_1
\exists X'_2 \forall y_2 \exists X_2
\ldots
\exists X'_n \forall y_n \exists X_n
$.  
Note that, for a set of variables $X$, we use the notation $\exists X$
to represent the existential quantification of the variables in $X$,
in any order (our discussion will always be independent of
any particular order chosen).
For $i \in [n]$, we refer to the variables in
$X'_i \cup \{ y_i \} \cup X_i$ as 
the \emph{level $i$ variables}.

\begin{itemized}

\item Define $B = \{ \neg x_{0,j,k} ~|~ j, k \in \{ 0, 1 \} \}
\cup \{ x_{n,j,0} \vee x_{n,j,1} ~|~ j \in \{ 0, 1 \} \}$.

\item For each $i \in [n]$ and each $j \in \{ 0, 1 \}$
define 
$H_{i,j} = \{ 
\neg x'_{i,0,k} \vee \neg x'_{i,1,l} \vee x_{i-1,j,0} \vee x_{i-1,j,1} 
~|~ k, l \in \{ 0, 1 \} \}$.

Observe that the clause
$\neg x'_{i,0,k} \vee \neg x'_{i,1,l} \vee x_{i-1,j,0} \vee x_{i-1,j,1}$
is logically equivalent to
$(x'_{i,0,k} \wedge x'_{i,1,l}) \rightarrow (x_{i-1,j,0} \vee x_{i-1,j,1})$.

\item For each $i \in [n]$,
define
$T_i = 
\{ \neg x_{i,0,k} \vee y_i \vee x'_{i,0,k} ~|~ k \in \{ 0, 1 \} \} 
\cup
\{ \neg x_{i,1,k} \vee \neg y_i \vee x'_{i,1,k} ~|~ k \in \{ 0, 1 \} \} 
$.

\end{itemized}

Define $\phi_{n}$
to be the
conjunction of the clauses contained in the just-defined sets.
Define $\Phi_n$ as $\prp_n : \phi_n$.
This definition of this family of sentences was inspired partially
by the separating formulas of~\cite{BEGJ00-relative,Ben-SassonImpagliazzoWigderson04-separation}.

Let us explain intuitively what the clauses mandate
and why the sentences $\Phi_n$ are false.
By the clauses in $B$, all of the variables $x_{0,j,k}$
must be set to $0$.  By the clauses in 
the sets $H_{1,j}$, 
either both variables $x'_{1,0,k}$ or both variables
$x'_{1,1,k}$ must be set to $0$.  Once this occurs,
the universal player can set the variable $y_1$
to $0$ or $1$ to force
either both variables $x_{1,0,k}$ or both variables
$x_{1,1,k}$ to $0$ (respectively), via the clauses in $T_1$.
This reasoning can then be repeated;
for instance, at the next level,
either both variables $x'_{2,0,k}$ or both variables
$x'_{2,1,k}$ must be set to $0$,
and then after
universal player assigning $y_2$ appropriately,
either both variables $x_{1,0,k}$ or both variables
$x_{1,1,k}$ are forced to $0$.  
In the end, the existential player must violate
one of the two clauses in $B$ concerning level $n$.

\begin{prop}
\label{prop:linear-size-qu-res-proofs}
The sentences $\{ \Phi_n \}_{n \geq 1}$ 
have QU-resolution proofs of size linear in $n$.
\end{prop}

Let $n \geq 1$;
we will use the following terminology to discuss $\Phi_n$.

We say that $r$ is a \emph{normal realization} of level $i \in [n]$
if it is an assignment defined on the level $i$ variables
such that, when $b$ is set to $r(y_i)$,
the following hold:
\begin{itemized}

\item 
$0 = r(x_{i,b,0}) = r(x'_{i,b,0}) = r(x_{i,b,1}) = r(x'_{i,b,1})$

\item
$r(x_{i,\neg b,0}) = r(x'_{i,\neg b,0}) \neq 
r(x_{i,\neg b,1}) = r(x'_{i,\neg b,1})$

\end{itemized}

We say that $r$ is a \emph{funny realization} of level $i \in [n]$
if it is an assignment defined on the level $i$ variables
such that, when $b$ is set to $r(y_i)$,
the following hold:
\begin{itemized}

\item
$r(x_{i,b,0}) = r(x'_{i,b,0}) \neq r(x_{i,b,1}) = r(x'_{i,b,1})$

\item 
$0 = r(x'_{i, \neg b, 0}) = r(x'_{i, \neg b, 1})$

\item 
$r(x_{i, \neg b, 0}) \neq r(x_{i, \neg b, 1})$

\end{itemized}

We state two key and straightforwardly verified properties
of realizations in the following proposition.

\begin{prop}
No assignment defined on the level $i$ variables
is both a normal realization and a funny realization.
Also, each normal realization and each funny realization
(of level $i$) satisfies all clauses in $T_i$.
\end{prop}

We define the set of assignments $F_n$ to be the set
containing all \emph{normal assignments} 
and all \emph{funny assignments}, which we now turn to define.
Let $f$ be a partial assignment of $\Phi_n$.
Let $\ell \geq 0$ denote the maximum level $\ell$
such that $f$ is defined on an $\exists$-variable
in level $\ell$.

We say that $f$ is a \emph{normal assignment}
if the following hold:
\begin{itemized}

\item $f$ is defined on the variables in 
$\{ x_{0,j,k} ~|~ j, k \in \{ 0, 1 \} \}$ 
and equal to $0$ on them.

\item For each $i \in [\ell - 1]$, 
the restriction of $f$ to the level $i$ variables
is a normal realization of level $i$.

\item If $\ell \geq 1$, either the restriction 
of $f$ to the level $\ell$ variables 
is a normal realization of level $\ell$;
or, $f$ is \emph{half-defined} on level $\ell$,
by which is meant that $f$ is not defined on any variables
in $\{ x_{\ell, j, k} ~|~ j, k \in \{ 0, 1 \} \}$,
but is defined on all variables
in $\{ x'_{\ell, j, k} ~|~ j, k \in \{ 0, 1 \} \}$
and has $\sum_{j, k \in \{ 0, 1 \}} x'_{\ell, j, k} = 1$.
\end{itemized}
For each normal assignment $f$, we define $s_n(f) = \ell$.

We say that $f$ is a \emph{funny assignment}
if there exists $m \in [\ell]$
such that the following hold:
\begin{itemized}

\item $f$ is defined on the variables in 
$\{ x_{0,j,k} ~|~ j, k \in \{ 0, 1 \} \}$ 
and equal to $0$ on them.

\item For each $i \in [m-1]$, 
the restriction of $f$ to the level $i$ variables
is a normal realization of level $i$.

\item The restriction of $f$ to the level $m$ variables
is a funny realization of level $m$.

\item For each $i$ with $m < i \leq \ell$ 
and for each $j \in \{ 0, 1 \}$, 
if $f$ is defined on one of the four variables 
in $\{ x_{i,j,k}, x'_{i,j,k} ~|~ k \in \{ 0, 1 \} \}$,
then it is defined on all of them
and $f(x_{i,j,0}) = f(x'_{i,j,0}) \neq f(x_{i,j,1}) = f(x'_{i,j,1})$.
\end{itemized}

\fullversion{
It is straightforward to verify that an assignment
cannot be both normal and funny, and also that,
if an assignment is funny, there exists a unique $m \in [\ell]$
witnessing this.
For each funny assignment $f$, we define $s_n(f) = m$.
We also identify the following properties of funny assignments which will be used.

\begin{prop}
\label{prop:funny-assignment}
Each funny assignment $f$ with $s_n(f) = m$ can be
extended to a funny assignment $f'$ with $s_n(f') = m$
which is defined on all $\exists$-variables.
Moreover, let $g$ be any assignment 
defined on all variables (of $\Phi_n$)
which extends a funny assignment $f'$
defined on all $\exists$-variables;
then, $g$ satisfies all clauses in $\phi_n$.
\end{prop}

\begin{theorem}
\label{thm:pair-satisfies-conditions}
For each $n \geq 1$, the pair
$(F_n, s_n)$ defined above satisfies the conditions
\emph{(semicompletion-of-empty)},
\emph{(monotonicity)},
\emph{($\forall$-branching)}, and
\emph{(double-branching)}
from Definition~\ref{def:delayer-strategy}.
\end{theorem}

\begin{proof}
We verify each of the conditions.

\emph{(semicompletion-of-empty)}:
The normal assignment $f$ defined only on 
$\{ x_{0,j,k} ~|~ j, k \in \{ 0, 1 \} \}$ 
is a semicompletion of the empty assignment
with $s(f) = 0$.

\emph{(monotonicity)}:
Suppose that $g \in F$.  Let $W \subseteq \dom(g)$.
We show that $g \res W$ has a semicompletion $f \in F$
with $s(f) \leq s(g)$. 

If the last variable $v$ in $W$  (according to $\prp_n$)
has the form $x'_{i,j,k}$ or $y_i$,
then set $f = g \res \{ z ~|~ z \preceq v \}$.
Otherwise, the last variable $v$ in $W$
has the form $x_{i,j,k}$,
and we set $f$ to be the restriction of $g$
to the variables in levels $0, \ldots, i$.

It is straightforward to verify that $f \in F$
and that  $s(f) \leq s(g)$.
We briefly indicate how to do so, as follows.
When $g$ is a normal assignment, then $f$ will also be 
a normal assignment.
When $g$ is a funny assignment with $s(g) = m$,
then consider two cases.
If the last variable $v$ in $W$ comes before 
or is equal to $y_m$, the assignment $f$ will be normal.
Otherwise, the assignment $f$ will be a funny assignment
with $s(f) = m$.

\emph{($\forall$-branching)}:
This property is straightforward to verify 
by examining the structure of the definitions
of \emph{normal assignment} and \emph{funny assignment}.
We perform the verification as follows.
Let $f \in F$ and let $y_i \notin \dom(f)$ 
be a universally quantified variable.
We claim that, for each $b \in \{ 0, 1 \}$,
the assignment $f[y_i \to b]$ is a semicompletion of $f$
having the same score as $f$.

When $f$ is a normal assignment,
we have $i \geq \ell$.  
We argue that $f[y_i \to b]$ is a normal assignment.
The assignments $f[y_i \to b]$ and $f$
are equal on variables in levels strictly before level $\ell$.
Also, 
if $f$ on level $\ell$ is a normal realization,
then $i > \ell$ and $f[y_i \to b]$ on level $\ell$
is the same normal realization.
If $f$ is half-defined on level $\ell$,
then $f[y_i \to b]$ is also half-defined on level $\ell$.
Thus, we have that $f[y_i \to b]$ is a normal assignment.
Clearly, 
$s(f[y_i \to b]) = s(f)$.

When $f$ is a funny assignment with $s(f) = m$, we have $i > m$.
In looking at the definition of a funny assignment
with a funny realization at level $m$,
the requirements imposed on the variables coming strictly 
after level $m$ concern only the existentially quantified variables.
Hence $f[y_i \to b]$ is also a funny assignment
with $s(f[y_i \to b]) = m$.

\emph{(double-branching)}:
Suppose that $f \in F$ and that $v \notin \dom(f)$.
We consider two cases.

When $f$ is a funny assignment with $s(f) = m$,
then the variable $v$ must occur in level $m+1$ or later.
If $v$ is a $\forall$-variable,
then set $b$ arbitrarily; we then have that
$g = f[v \to b]$ is a funny assignment with $s(g) = s(f) = m$.
If $v$ is an $\exists$-variable,
then it is of the form $x_{i,j,\ell}$ or $x'_{i,j,\ell}$;
take $g$ to be the either of the two extensions of $f$
defined on 
$\dom(f) \cup \{ x_{i,j,k}, x'_{i,j,k} ~|~ k \in \{ 0, 1 \} \}$
with 
$f(x_{i,j,0}) = f(x'_{i,j,0}) \neq f(x_{i,j,1}) = f(x'_{i,j,1})$.
We have that $g$ is a funny assignment with $s(g) = s(f) = m$.

When $f$ is a normal assignment with $s(f) = \ell$,
the variable $v$ must come after all variables in $\dom(f)$.
We may assume that $v$ is an $\exists$-variable
(otherwise, one may reason as in the case of
the condition \emph{($\forall$-branching)}  
to obtain a semicompletion $g$ with $s(g) = s(f)$.)

First, suppose that 
the restriction of $f$ to level $\ell$ is
a normal realization.
If $v$ is in level $\ell + 1$, 
then both $f[v \to b]$ and $f[v \to \neg b]$
have semicompletions $g$ and $g'$ (respectively)
which are defined on levels $0$ through $\ell + 1$ inclusive
and are equal to funny realizations on level $\ell + 1$.
In this situation, we have $s(g) = s(g') = \ell + 1$.
If $v$ is in level $\ell + 2$ or a later level,
then we can obtain semicompletions $g$ and $g'$
of $f[v \to b]$ and $f[v \to \neg b]$ (respectively)
as follows.
First, extend $f$ to obtain an assignment that is
 equal to an arbitrary 
funny realization on level $\ell + 1$;
then, we may extend the result 
by reasoning as in the previous case
(where $f$ is a funny assignment) to obtain
the desired semicompletions $g$ and $g'$, which 
are funny assignments with $s(g) = s(g') = \ell + 1$.

Next, suppose that $f$ is half-defined on level $\ell$.
\begin{itemize}
\item[(a)] If the variable $v$ has the form $x_{\ell, j, k}$,
then we extend $f$ as follows: set $y_{\ell}$ 
arbitrarily if it is not already defined,
and then extend the result
so that it is either a normal realization
or a funny realization at level $\ell$.
The resulting assignment $g$ is a semicompletion of
$f[v \to g(v)]$ where $s(g) = s(f)$.

\item[(b)] If the variable $v$ has the form
$x'_{\ell + 1, j', k'}$, then take the assignment
from the previous item (a); for each value of $b \in \{ 0, 1 \}$,
this assignment can be extended 
to be defined on the variables $x'_{\ell+1,j,k}$
so that 
it is equal to $b$ on $x'_{\ell + 1, j', k'}$
and $\sum_{j, k \in \{ 0, 1 \}} x'_{\ell+1,j,k} = 1$.
The resulting extensions are semicompletions
of $f[v \to b]$ and $f[v \to \neg b]$ with score $\ell + 1$.

\item[(c)] Otherwise,
 the variable $v$ has the form $x_{\ell+1,j',k'}$
 or occurs at level $\ell+2$ or later.
 We first take the extension $h$ of $f$ that 
is described in item (a); $h$ is defined on
 all variables in level $\ell$.
 If $h$ is equal to a funny realization at level $\ell$,
 then pick $b$ arbitrarily; 
 it is straightforwardly verified that
 there is a funny assignment $g$ that is a semicompletion
 of $h[v \to b]$; this assignment $g$ has $s(g) = \ell = s(f)$.
 If $h$ is equal to a normal realization at level $\ell$,
 then one can straightforwardly verified that
 both $h[v \to b]$ and $h[v \to \neg b]$ have
 semicompletions $g$ and $g'$ (respectively)
 which are funny assignments having
 $s(g) = s(g') = \ell + 1$.

\end{itemize}

\end{proof}

\begin{lemma}
\label{lemma:many-points}
Let $d, n \in \nats$ be such that $2 \leq d \leq 2n$.
Each assignment $f \in F_n \cap H(\Phi_n,d)$
has $s(f) > n - \lceil d/2 \rceil$.
\end{lemma}

\begin{proof}
Suppose that $f \in F_n \cap H(\Phi_n, d)$.
It cannot be that $f$ is a funny assignment,
as for any funny assignment $f$, 
the QBF $\Phi_n[f]$ is true as a consequence of  
Proposition~\ref{prop:funny-assignment}.
Thus $f$ is a normal assignment.
Suppose, for a contradiction, that
$s(f) \leq n - \lceil d/2 \rceil$.
In this case, $f$ is not defined
on any of the $\exists$-variables in 
the last $\lceil d/2 \rceil$ levels,
that is, $f$ is not defined
on any of the $\exists$-variables in
levels $n-(\lceil d/2 \rceil + 1), \ldots, n-1, n$.
As a consequence, the prefix of $\Phi_n[f]$
is not $\Pi_d$.
Now, consider the relaxation $\Phi' = \prpp:\phi'$
of $\Phi_n[f]$ witnessing
that $f \in H(\Phi_n,d)$.
Since $\prpp$ is $\Pi_d$,
it must hold that, in $\prpp$,
there exists a level 
$m \in \{ n-(\lceil d/2 \rceil + 1), \ldots, n-1, n \}$
such that the variable $y_i$ comes before the variables
in $X'_i$.

We prove that $\Phi'$ is true
(this suffices, as it contradicts $f \in H(\Phi_n,d)$).
We describe an $\exists$-winning strategy for
$\Phi'$, as follows.
After each level is set, the resulting assignment
is in $F_n$.
When it is time to set an $\exists$-variable in level $i$,
first check if it holds that $i = m$ and
no previous level is set to a funny realization.
If these two conditions hold, then level $i = m$
is set to a funny realization.
Otherwise, the variables at level $i$ are set as follows.
\begin{itemize}

\item 
If a previous level 
is set to a funny realization, then the variables in
$X'_i \cup X_i$ are set so that the resulting assignment
remains funny (this can in fact be done without looking
at the value of $y_i$).

\item Otherwise, proceed as follows.  
The variables in $X'_i$ are set so that the sum of their values
is equal to $1$.
The variables in $X_i$ are set so that, at level $i$,
one obtains either a normal or funny realization.

\end{itemize}
This $\exists$-strategy is winning, 
as no matter how the universal player plays,
the end assignment will be a funny assignment
(defined on all variables), 
which satisfies all clauses 
(Proposition~\ref{prop:funny-assignment}).
\end{proof}

By Theorem~\ref{thm:pair-satisfies-conditions}
and Lemma~\ref{lemma:many-points}, in conjunction with
Theorem~\ref{thm:delayer-many-leaves}
and Lemma~\ref{lemma:semicompletion-closure},
we obtain the following result.
}

\confversion{The following result is obtained by
applying the main theorem of the previous section
to the strategies $(F_n, s_n)$.
}

\begin{theorem}
\label{thm:lower-bound-tree-like-relaxing-qu-res}
Tree-like relaxing QU-res
requires proofs of size $\Omega(2^n)$
on the sentences $\{ \Phi_n \}_{n \geq 1}$.
\end{theorem}


\section{Lower bound for \emph{relaxing QU-res}}
\label{sect:lower-bound-relaxing-qu-res}

We define a family of QBCs, to be studied in this section,
as follows.
Let $n \geq 1$.
Define $\prp_n$ to be the quantifier prefix
$\exists x_1 \forall y_1 \ldots \exists x_n \forall y_n$.
Define $\phi_{n,j}$
to be
true if and only if
$j + \sum_{i=1}^n (x_i + y_i) \not\equiv n \pmod 3$.
Define $\Phi_{n}$ to be the sentence
$\prp_n : \phi_{n,0}$; these are the sentences that will be used
to prove the lower bound.
It is straightforward to verify that $\phi_n$
can be represented as a circuit of size polynomial in $n$,
and we assume that $\phi_n$ is so represented.
\fullversion{We will also make use of the QBCs defined by
$\Phi_{n,j} = \prp_n : \phi_{n,j}$.}

\begin{prop}
For each $n \geq 1$, the sentence $\Phi_n$ is false.
\end{prop}

It is straightforward to verify that a winning $\forall$-strategy
is to set the variable $y_i$ to the value $\neg x_i$.

To obtain the lower bound, we show that
for any proof $\pi$, the graph $G(\pi)$
must have exponentially many sinks.
We begin by showing that any assignment
to an initial segment of the $\exists$-variables
can be mapped naturally to a sink.

\begin{lemma}
\label{lemma:mod3game-leaf-agrees}
Let $\pi$ be a relaxing QU-res proof of $\Phi_n$
from an axiom set, and
suppose $t \geq 1$.
Let 
$f: \{ x_1, \ldots, x_{n - \lceil t/2 \rceil} \}
\to \{ 0, 1 \}$ 
be an assignment.
There exists a sink of $G(\pi)$ whose label agrees with $f$.
\end{lemma}

We next show that each sink must be defined on a
variable that occurs \emph{towards the end} of
the quantifier prefix, made precise as follows.

\begin{lemma}
\label{lemma:sink-label-defined-at-end}
Let $\pi$ be a relaxing QU-res proof of $\Phi_n$
from axiom set $H(\Phi, \Pi_t)$,
where $t \geq 2$ and $n \geq \lceil t/2 \rceil$.
Each sink of $G(\pi)$
has a label $a$ that is defined on 
one of the following variables:
\begin{center}
$x_{n-(\lceil t/2 \rceil - 1)}, y_{n-(\lceil t/2 \rceil - 1)},
\ldots, x_{n-1}, y_{n-1}, x_n, y_n$.
\end{center}
\end{lemma}

\fullversion{
\begin{proof}
Suppose that there exists a sink of $G(\pi)$ with label $a$
that is not defined on one of the specified variables.
We show that any $\Pi_t$-relaxation 
of $\Phi_{n}[a]$ is true, to obtain a contradiction.
It suffices to prove that, for \emph{any} assignment
$f: \{ x_1, y_1, \ldots, 
x_{n- \lceil t/2 \rceil}, y_{n- \lceil t/2 \rceil} \} \to \{ 0, 1 \}$,
any $\Pi_t$-relaxation of $\Phi_{n}[f]$ is true.
The sentence $\Phi_{n}[f]$ is truth-equivalent to a sentence of the form
$\Phi_{\lceil t/2 \rceil, j}$.
This latter sentence has an even number of variables
which number is
greater than or equal to $t$, and is not $\Pi_t$.
Now consider a $\Pi_t$-relaxation 
$\prp:\phi_{\lceil t/2 \rceil, j}$ of
$\Phi_{\lceil t/2 \rceil, j}$.
We claim that this relaxation 
$\prp:\phi_{\lceil t/2 \rceil, j}$ is true.
Since $\prp$ is a $\Pi_t$-relaxation of the prefix of 
$\Phi_{\lceil t/2 \rceil, j}$,
there exists a variable $x_k$
such that $y_k$ appears to its left in $\prp$.
We describe a winning $\exists$-strategy that witnesses
the truth of the relaxation.
First, consider the case that $k = 1$.
In this case, the variables $x_1$ and $x_2$ can be set
so that $y_1 + x_1 + x_2  \equiv -j \pmod{3}$,
and each other variable $x_i$ can be set to be not equal to
$y_{i-1}$.  Then, no matter how the universal variables are set,
the sum of all of the variables excluding $y_n$
will be $y_1 + x_1 + x_2 + (y_2 + x_3) + \cdots + (y_{n-1} + x_n)
\equiv y_1 + x_1 + x_2 + (n-2) \equiv n+1-j \pmod{3}$.
Then, no matter how $y_n$ is set, the final sum $S$ of the variables
will have $S \not\equiv n-j \pmod{3}$.
In the case that $k \neq 1$,
the variable $x_1$ is set arbitrarily,
and the variables $x_k$ and $x_{k+1}$
are set so that 
$x_1 + y_{k-1} + y_k + x_k + x_{k+1} 
\equiv 1-j \pmod{3}$.
Each other variable $x_i$ is set to be not equal to 
$y_{i-1}$.  
No matter how the universal variables are set,
the sum of all of the variables excluding $y_n$
will be
$x_1 + y_{k-1} + y_k + x_k + x_{k+1} + (n-3) 
\equiv n+1-j \pmod{3}$,
which is sufficient as in the previous case.
%
\end{proof}
}

When $f$ is a partial assignment of $\Phi_n$,
we refer to the elements of
$\{ v ~|~ v \preceq \last(f) \} \setminus \dom(f)$
as \emph{holes}.

\begin{lemma}
\label{lemma:at-most-one-hole}
Let $\pi$ be a relaxing QU-res proof of $\Phi_n$
from an axiom set of the form $H(\Phi_n, \Pi_t)$.
Each sink of $G(\pi)$ has a label $f$ having at most one hole.
\end{lemma}

\fullversion{
\begin{proof}
Suppose that $f$ has $2$ or more holes.
There exist extensions $f_0$, $f_1$, and $f_2$
defined on $V = \{ v ~|~ v \preceq \last(f) \}$
such that $\sum_{v \in V} f_i(v) \equiv i \pmod{3}$
for each $i = 0, 1, 2$.
It is straightforward to verify that one of
the QBCs $\Phi_n[f_0]$, $\Phi_n[f_1]$, $\Phi_n[f_2]$ 
is true, implying the truth of $\Phi_n[f]$
and contradicting that $\clause(f)$ is an axiom
in $H(\Phi_n, \Pi_t)$.
\end{proof}
}

\begin{theorem}
Suppose that $t \geq 2$ and that $n \geq \lceil t/2 \rceil$.
Let $\pi$ be a QU resolution proof of $\Phi_n$
from the axiom set $H(\Phi_n, \Pi_t)$.
The graph $G(\pi)$ has at least $2^{n - \lceil t/2 \rceil - 1}$
sinks.
\end{theorem}

\begin{proof}
By Lemma~\ref{lemma:mod3game-leaf-agrees},
for each assignment
$f: \{ x_1, \ldots, x_{n - \lceil t/2 \rceil} \}
\to \{ 0, 1 \}$,
there exists a sink $v_f$ of $G(\pi)$ whose 
label agrees with $f$.
For each label $g$ of each sink, any variable
in $\{ x_1, \ldots, x_{n - \lceil t/2 \rceil} \}$
on which $g$ is not defined must be a hole of $g$,
by Lemma~\ref{lemma:sink-label-defined-at-end}.

Fix a mapping taking each such assignment $f$
to such a sink $v_f$.
Since the label of each sink has at most $1$ hole
(by Lemma~\ref{lemma:at-most-one-hole}),
each sink is mapped to by at most two assignments.
Hence the number of sinks must be at least the
number of assignments of the form
$f: \{ x_1, \ldots, x_{n - \lceil t/2 \rceil} \}
\to \{ 0, 1 \}$ divided by two.
\end{proof}

From the previous theorem, we immediately obtain the following.

\begin{theorem}
\label{thm:lower-bound-relaxing-qu-res}
Relaxing QU-res
requires proofs of size $\Omega(2^n)$
on the sentences $\{ \Phi_n \}_{n \geq 1}$.
\end{theorem}

\fullversion{
\section{Discussion}

Beyond the proof systems discussed already in the paper,
another natural way to certify the falsity of a QBF
is by explicitly representing a winning $\forall$-strategy.
Sometimes, the QBF literature refers to 
methods for extracting strategies
from falsity proofs
or by outfitting a solver; this notion is often called
\emph{strategy extraction}.

We can formalize a proof system ensemble based on 
explicit representation of $\forall$-strategies, as follows.
We use the notation in Section~\ref{sect:preliminaries}.
Let $\Phi$ be a QBC, and let $H$ be an axiom set of $\Phi$.
Let us define a \emph{circuit $\forall$-strategy}
to be a sequence of circuits $(C_y)_{y \in Y}$
where each $C_y$ has $|X_{<y}|$ input gates,
which are labelled with the elements of $X_{<y}$.
Such a sequence $(C_y)_{y \in Y}$
 naturally induces a $\forall$-strategy
$(\tau_y)_{y \in Y}$ for $\Phi$.
We say that  $(C_y)_{y \in Y}$
is
a \emph{winning circuit $\forall$-strategy
with respect to $H$}
if for every assignment $\sigma: X \to \{ 0, 1 \}$,
it holds that $\langle \tau, \sigma \rangle$
falsifies a clause in $H$.
This naturally yields a proof system ensemble $(A, r)$,
where $r(k) = k+4$
and $A$ accepts $(k,(\Phi,\pi))$
when the following condition holds:
$\pi$ is a winning circuit $\forall$-strategy for $\Phi$
with respect to $H(\Phi, \Pi_{k+2})$, that is, 
if for each assignment $\sigma: X \to \{ 0, 1 \}$,
there exists a clause $C \in H(\Phi,\Pi_{k+2})$
such that $\langle \tau, \sigma \rangle$ falsifies $C$.
The latter formulation of the condition can be checked
with access to a $\Pi^p_{k+4}$ oracle (equivalently, 
a $\Sigma^p_{k+4}$ oracle).
We call this proof system ensemble \emph{relaxing stratex}.

From a result appearing in previous work~\cite[Section 3.1]{GoultiaevaVanGelderBacchus11-uniform}, 
it can be shown that
winning circuit $\forall$-strategies 
can be efficiently computed from QU-resolution proofs.
This implies the following.
\begin{prop}
(derivable from~\cite[Section 3.1]{GoultiaevaVanGelderBacchus11-uniform})
Relaxing stratex effectively simulates relaxing QU-res.
\end{prop}
The QBC family studied in the previous section
had very simple winning $\forall$-strategies
which can clearly be represented by polynomial-size circuits.
We can thus conclude from Theorem~\ref{thm:lower-bound-relaxing-qu-res}
that relaxing QU-res does not simulate relaxing stratex.
The separation between tree-like relaxing QU-res and 
(general)  relaxing QU-res
(Proposition~\ref{prop:linear-size-qu-res-proofs}
and Theorem~\ref{thm:lower-bound-tree-like-relaxing-qu-res})
implies that tree-like relaxing QU-res 
does not simulate relaxing QU-res,
while it is clear that relaxing QU-res simulates 
tree-like relaxing QU-res.
The technical results under discussion can thus be summarized
via a small hierarchy of proof system ensembles:
tree-like relaxing QU-res is simulable by relaxing QU-res, 
but not the other way around;
and, relaxing QU-res is simulable by relaxing stratex,
but not the other way around.
}


%% file: appendix.tex
\fullversion{
\newpage

\appendix

\section{Proof of Proposition~\ref{prop:polybounded-iff-in-ph}}

\begin{proof}
For the forward direction, let $(A, r)$ be
a polynomially bounded proof system ensemble for $L$.
Let $k$ and $f$ be as in Definition~\ref{def:polybounded}.
Fix $k' = r(k)$.
Let $q$ be a polynomially such that for each string $x$,
it holds that $|f(x)| \leq q(|x|)$.
Membership of a given string $x$ in $L$ 
can be decided by nondeterministically guessing a string $\pi$
of length less than or equal to $q(|x|)$,
and then checking if $A^{S(k')}$ accepts; 
this places $L$ in the polynomial hierarchy.

For the backward direction, suppose that $L$ is in the PH.
Then, there exists a $k' \in \nats$
and a polynomial time algorithm $B$ which may make
oracle calls to $S(k')$ such that
$B^{S(k')}$ accepts a string $(x, \pi)$ 
if and only if $x \in L$.
Define $r$ to map each $n \in \nats$ to $k'$.
The pair $(B, r)$ is readily verified to be a proof system
ensemble which is polynomially bounded
(indeed, 
with respect to the constant polynomial equal everywhere to $1$).
\end{proof}

\section{Proof of Proposition~\ref{prop:relaxing-qu-res-an-ensemble}}
\begin{proof}
Suppose first that $(k, (\Phi, \pi))$
is accepted by $A$.
Then, $\pi$ is a QU-resolution falsity proof of
$\Phi$ from axioms in $H(\Phi, \Pi_{k+2})$;
by Proposition~\ref{prop:relaxations-give-axiom-set},
$H(\Phi, \Pi_{k+2})$ is an axiom set of $\Phi$,
so it follows from Proposition~\ref{prop:soundness-qu-res}
that the QBC $\Phi$ is false.

Suppose that $\Phi = \prp : \phi$
is a false QBC; let $V$ denote its variables.
Let $F$ be the set that contains each assignment
$f: V \to \{ 0, 1 \}$ that falsifies $\phi$.
We have that $\phi$ has the same
satisfying assignments as $\phi' = \bigwedge_{f \in F} \clause(f)$.
Hence the QBC $\Phi' = \prp : \phi'$
is also false.
It is known that there exists a QU-resolution proof $\pi$
of $\Phi'$ ending with the empty clause, 
from axiom set $H_F = \{ \clause(f) ~|~ f \in F \}$;
this follows from the \emph{completeness} of Q-resolution
on clausal QBF.
Since $H_F \subseteq H(\Phi, \Pi_{k+2})$ for each $k$,
it holds that $A$ accepts $(k, (\Phi, \pi))$.
\end{proof}

\section{Proof of Lemma~\ref{lemma:semicompletion-closure}}

\begin{proof}
Suppose that $f \in H(\Phi, \Pi_m)$ and that
$g$ is a semicompletion of $f$.
Suppose that $v \in \dom(g) \setminus \dom(f)$.
Assume that $\dom(f)$ is non-empty.
If $\last(f)$ is a $\forall$-variable,
then by definition of semicompletion,
either $v \precneq \last(f)$;
or, each variable of the quantifier block of 
$\last(f)$ is in $\dom(f)$ and $v$ occurs
in the quantifier block (of $\exists$-variables)
immediately following the block of $\last(f)$.
If $\last(f)$ is an $\exists$-variable,
then by definition of semicompletion,
it holds that $v \preceq \last(f)$.
In each of these cases and also when $\dom(f) = \emptyset$, 
it holds that $v$ is in the first 
quantifier block of $\Phi[f]$,
which block is existentially quantified.

We have thus established that each variable
in $\dom(g) \setminus \dom(f)$ is existentially quantified
in $\Phi[f]$.
Let $\Phi' = \prpp:\phi'$
be a false $\Pi_m$-relaxation of $\Phi[f]$.
Let $\Phi''$ be the sentence obtained from $\Phi$
by replacing each variable $v \in \dom(g) \setminus \dom(f)$
with the constant $g(v)$ in $\phi'$,
and removing each such $v$ (and its accompanying quantifier)
from $\prpp$.
We have that $\Pi''$ is a $\Pi_m$-relaxation of
$\Phi[g]$,
and that the falsity of $\Phi'$ implies the falsity of $\Phi''$.
\end{proof}

\section{Proof of Proposition~\ref{prop:linear-size-qu-res-proofs}}

\begin{proof}
We prove, by induction, that for $c = 0, \ldots, n$,
it holds that, for each $j \in \{ 0, 1 \}$,
the clause $x_{n-c,j,0} \vee x_{n-c,j,1}$
is derivable from $\Phi_n$ by QU-resolution.
For $c = 0$, we have that the two clauses of concern
are contained in $B$.
Suppose that $c \in [n]$ and that the claim is true for $c-1$.
By induction, we have that the two clauses
$D_0 = x_{n-(c-1),0,0} \vee x_{n-(c-1),0,1}$ and
$D_1 = x_{n-(c-1),1,0} \vee x_{n-(c-1),1,1}$
are derivable by QU-resolution.
By resolving the clause $D_0$ with the two clauses
in 
$$
\{ \neg x_{n-(c-1),0,k} \vee \neg y_i \vee x'_{n-(c-1),0,k} ~|~ k \in \{ 0, 1 \} \} \subseteq T_{n-(c-1)}
$$
we derive the clause 
$\neg y_i \vee 
x'_{n-(c-1),0,0} \vee
x'_{n-(c-1),0,1}$;
by applying $\forall$-elimination,
we derive the clause
$E_0 = x'_{n-(c-1),0,0} \vee  x'_{n-(c-1),0,1}$.
Similarly, 
by resolving the clause $D_1$ with the two clauses
in
$$
\{ \neg x_{n-(c-1),1,k} \vee y_i \vee x'_{n-(c-1),1,k} ~|~ k \in \{ 0, 1 \} \} 
\subseteq T_{n-(c-1)}$$
we derive the clause
$y_i \vee x'_{n-(c-1),1,0} \vee x'_{n-(c-1),1,1}$;
by applying $\forall$-elimination,
we derive the clause 
$E_1 = x'_{n-(c-1),1,0} \vee x'_{n-(c-1),1,1}$.
By resolving $E_0$ and $E_1$ with the clauses in
$H_{n-(c-1),0}$, we derive the clause
$x_{n-c,0,0} \vee x_{n-c,0,1}$.
Similarly,
by resolving $E_0$ and $E_1$ with the clauses in
$H_{n-(c-1),1}$, we derive the clause
$x_{n-c,1,0} \vee x_{n-c,1,1}$.  This concludes the proof of the claim.

The empty clause is obtained by resolving the unit clauses
$\{ \neg x_{0,j,k} ~|~ j, k \in \{ 0, 1 \} \} \subseteq B$
with the clause $x_{0,0,0} \vee x_{0,0,1}$,
or with the clause $x_{0,1,0} \vee x_{0,1,1}$.
The resulting proof has linear size, since each step of the induction
requires a constant amount of size.
\end{proof}

\section{Proof of Lemma~\ref{lemma:mod3game-leaf-agrees}}
\begin{proof}
Fix such an assignment $f$.
Since the empty assignment, the label of the root,
agrees with $f$, it suffices to show the following:
each non-leaf node $u$ whose label agrees with $f$,
has an edge to a node which agrees with $f$.
If $u$ has one outgoing edge, then this is clear
by the description in Proposition~\ref{prop:proof-graph}.
If $u$ has two outgoing edges, let $v$ be the variable
described in Proposition~\ref{prop:proof-graph}.
If $v$
is universally quantified, then both of its children agree with $f$;
if $v$ is existentially quantified, then one of the children
must have label $a$ where $a(v) = f(v)$; 
this child's label $a$ then agrees with $f$.
\end{proof}
}